\documentclass[11pt]{article}
\usepackage{amsmath}
\usepackage{amsfonts}
\usepackage{amssymb}
\usepackage{amsthm}
\usepackage{thmtools, thm-restate}
\usepackage{fullpage}
\usepackage{comment}
\usepackage{graphicx}
\usepackage{hyperref}
\usepackage{lmodern}
\usepackage{caption}
\usepackage{lipsum}
\usepackage{graphicx}
\usepackage{subcaption}
\usepackage{appendix}
\usepackage[normalem]{ulem}

\usepackage[
    backend=biber,
    style=alphabetic,
    sorting=nyt,
    maxbibnames=99 
]{biblatex}

\DeclareMathAlphabet{\mymathbb}{U}{bbold}{m}{n}

\declaretheorem[{style=definition,numberwithin=section}]{definition}
\declaretheorem[{style=definition,sibling=definition}]{theorem}
\declaretheorem[{style=definition,sibling=definition}]{lemma}

\declaretheorem[{style=definition,sibling=definition}]{corollary}
\declaretheorem[{style=definition,sibling=definition}]{proposition}

\declaretheorem[{style=definition,sibling=definition}]{remark}

\DeclareMathOperator{\supp}{supp}

\DeclareMathOperator{\wt}{wt}

\newcommand{\avgq}{\mathrm{D}_{\mathrm{ave}}}

\newcommand{\bs}{\mathrm{bs}}
\newcommand{\cert}{\mathrm{C}}

\newcommand{\ac}{\mathsf{AC}^0}
\newcommand{\detq}{\mathrm{D}}
\newcommand{\randq}{\mathrm{R}}

\newcommand{\cost}{\mathrm{cost}}
\newcommand{\bnm}{\mathcal{B}_{n,m}}
\newcommand{\prandrtrc}{\mathcal{R}_p}
\newcommand{\bitset}{\{0,1\}}

\newcommand{\bfand}{\mathrm{AND}}
\newcommand{\bfor}{\mathrm{OR}}
\newcommand{\bfxor}{\mathrm{XOR}}

\newcommand{\pso}{\mathrm{PSO}}
\newcommand{\dtsize}{\mathrm{DT}_\mathrm{size}}

\newcommand{\maj}{\mathrm{Maj}}
\newcommand{\var}{\mathrm{{Var}}}
\newcommand{\influ}{\mathrm{{Inf}}}

\newtheorem*{theorem*}{Theorem}
\newtheoremstyle{named}{}{}{\itshape}{}{\bfseries}{.}{.5em}{\thmnote{#3}#1}
\theoremstyle{named}

\providecommand{\keywords}[1]
{
  \textbf{\textit{Keywords---}} #1
}

\title{
    Average-Case Deterministic Query Complexity \\ of Boolean Functions with Fixed Weight 
}
\author{
    Yuan Li \\
    \texttt{yuan\_li@fudan.edu.cn} \\
    Fudan University \\
    \and 
    Haowei Wu \\
    \texttt{hwwu21@m.fudan.edu.cn} \\
    Fudan University \\
    \and
    Yi Yang \\
    \texttt{yyang1@fudan.edu.cn} \\
    Anqing Normal University\\
    Fudan University
}

\date{}

\begin{document}

\maketitle

\begin{abstract}
    We study the \textit{average-case deterministic query complexity} of boolean functions under a \emph{uniform input distribution}, denoted by  $\avgq(f)$,  the minimum average depth of zero-error decision trees that compute a boolean function $f$. This measure has found several applications across diverse fields, yet its understanding is limited.
    
    We study boolean functions with fixed weight, where weight is defined as the number of inputs on which the output is $1$.
    We prove \(\avgq(f) \le \max \left\{ \log \frac{\wt(f)}{\log n} + O(\log \log \frac{\wt(f)}{\log n}), O(1) \right\}\) for every $n$-variable boolean function $f$, where $\wt(f)$ denotes the weight. For any $4\log n \le m(n) \le 2^{n-1}$, we prove the upper bound is tight up to an additive logarithmic term for almost all $n$-variable boolean functions with fixed weight $\wt(f) = m(n)$. 
    
    H{\aa}stad's switching lemma or Rossman's switching lemma [Comput. Complexity Conf. 137, 2019] implies $\avgq(f) \leq  n\left(1 - \frac{1}{O(w)}\right)$ or $\avgq(f) \le n\left(1 - \frac{1}{O(\log s)}\right)$ for CNF/DNF formulas of width $w$ or size $s$, respectively. 
    We show there exists a DNF formula of width $w$ and size $\lceil 2^w / w \rceil$ such that $\avgq(f) = n \left(1 - \frac{\log n}{\Theta(w)}\right)$ for any $w \ge 2\log n$.
\end{abstract}


\section{Introduction}

The \textit{average-case deterministic  query complexity} of a boolean function $f$ under a \textit{uniform input distribution}\footnote{In \cite{AD01}, the complexity under distribution $\mu$ is denoted by $\mathrm{D}^\mu (f)$, and $\mu$ could be arbitrary. In this paper, we assume the input distribution $\mu$ is uniform.}, denoted by $\avgq(f)$, 
is the minimum average depth of zero-error decision trees that compute $f$.
This notion serves as a natural average-case analogy of the classic deterministic query complexity $\detq(f)$ and has found applications in query complexity, boolean function analysis, learning algorithms, game theory, and percolation theory. 
Besides that, $\avgq(f)$ is a measure with limited understanding, since $\avgq(f)$ falls outside the class of polynomially-related measures, which includes $\detq(f)$, $\randq(f)$, $\cert(f)$, $\bs(f)$, and $\mathrm{s}(f)$ (see  the summaries in \cite{BD02, ABK16, ABK+21} and Huang's proof of the Sensitivity Conjecture \cite{Hua19}). 
This work is also inspired by Rossman's circuit lower bounds of detecting $k$-clique on Erdős–Rényi graphs in the average case \cite{Ros08, Ros14}. Through this paper, we hope to initiate a comprehensive study on $\avgq(f)$, exploring its implications and applications.

\subsection{Background}

To our knowledge, Ambainis and de Wolf were the first to introduce the concept of \emph{average-case} query complexity \cite{AD01}. They showed super-polynomial gaps between average-case deterministic query complexity, average-case bounded-error randomized query complexity, and average-case quantum query complexity. 

Prior to the conceptualization by Ambainis and de Wolf, average-case query complexity had been studied implicitly since the early days of computer science. Yao \cite{Yao77} noticed that $\avgq^\mu(f)$ (with respect to any distribution $\mu$) lower bounds the zero-error randomized query complexity, i.e., $\avgq^\mu(f) \le \randq_0(f)$.
Furthermore, Yao's minimax principle says the maximum value of $\avgq^\mu(f)$ over all distributions $\mu$ equals $\randq_0(f)$.

O'Donnell, Saks, Schramm, and Servedio  established the OSSS inequality \cite{OSS+05, Lee10, JZ11}: $\avgq^{\mu_p}(f) \ge \frac{\var[f]}{\max_i \influ_i[f]}$ for any boolean function $f$, where $\mu_p$ is the $p$-biased distribution and $\influ_i[f]$ is the influence of coordinate $i$.
By applying the inequality, O'Donnell et al. \cite{OSS+05} showed that $\randq_0 (f) \geq \avgq^{\mu_p}(f) \ge (n/\sqrt{4p(1 - p)})^{2/3}$ for any nontrivial monotone $n$-vertex graph property $f$ with critical probability $p$.
This result made progress on Yao's conjecture \cite{Yao77}, which asserts that $\randq_0(f) = \Omega(n)$ for every nontrivial monotone graph property. 
When $p = 1/2$, we have $\randq_0(f) \ge \avgq(f) \ge n^{2/3}$; Benjamini et al. proved that the lower bound $\avgq(f) \ge n^{2/3}$ is almost tight \cite{BSW05}.

While studying learning algorithms, O'Donnell and Servedio \cite{OS06} discovered the OS inequality: $ (\sum_{i} \hat{f}(\{i\}))^2 \le  \avgq(f) \le \log \dtsize(f) $ for any boolean function $f$, where $\hat{f}(\cdot)$ denotes $f$'s Fourier coefficient and $\dtsize(f)$ denotes the decision tree size. The OS inequality plays a crucial role in learning monotone boolean functions (under the uniform distribution).

The most surprising connection (application) arose in game theory. Peres, Schramm, Sheffield, and Wilson \cite{PSS+07} studied the \emph{random-turn} HEX game, in which two players determine who plays next by tossing a coin before each round. They proved that the expected playing time (under optimal play) coincides with $\avgq(f)$, where $f$ is the $L \times L$ hexagonal cells connectivity function. Using the OS inequality and the results of Smirnov and Werner on percolation \cite{SW01}, Peres et al. proved a lower bound $L^{1.5 + o(1)}$ on the expected playing time on an $L \times L$ board.

\subsection{Our results}

The \emph{weight} of a boolean function, defined as the number of inputs on which the output is $1$, is related to its query complexity. For instance, Ambainis et al. \cite{Amb+16} proved that the quantum query complexity of almost all $n$-variable functions with fixed weight $m$ is $\Theta\left(\frac{\log m}{c + \log n - \log \log m} + \sqrt n\right)$, where $c > 0$ is a constant. 
In contrast, the hardest function with weight $m$ has quantum query complexity $\Theta\left(\left({n\cdot\frac{\log m}{c + \log n - \log \log m}}\right)^{1/2} + \sqrt n\right)$. 
Ambainis et al. \cite{Amb+16} also proved that almost all functions with fixed weight $m \ge 1$ have randomized query complexity $\Theta(n)$ as the hardest one. 

Our first result proves that $\avgq(f) \le \log \frac{m}{\log n} + O(\log \log \frac{m}{\log n})$ for any $n$-variable boolean function $f$ with weight $m \ge 4 \log n$.

\begin{theorem} \label{thm:intro_bnm_ub}
    For every boolean function $f : \bitset^n \rightarrow \bitset$, if the weight $\wt(f) \ge 4\log n$, then
    \begin{equation}
        \label{equ:main_thm_ub}
        \avgq(f)  \leq  \log \frac{\wt(f)}{\log n} + O\left(\log \log \frac{\wt(f)}{\log n}\right).
    \end{equation}
    Otherwise, $\avgq(f) = O(1)$. 
\end{theorem}

We prove Theorem \ref{thm:intro_bnm_ub} by designing a recursive query algorithm that attains the query complexity given in \eqref{equ:main_thm_ub}. The algorithm queries an arbitrary bit until the subfunction's weight becomes sufficiently small, or more specifically, smaller than the logarithm of its input length; once this border condition is met, we invoke another algorithm which, on average, takes $O(1)$ bits to query the subfunction.

Next, we prove Theorem \ref{thm:intro_bnm_lb}, complementing our first result, which says that Theorem \ref{thm:intro_bnm_ub} is tight up to a lower order term for almost all fixed-weight functions.  

\begin{theorem} \label{thm:intro_bnm_lb} 
    Let $m:\mathbb{N} \to \mathbb{N}$ be a function such that $4\log n \le m(n) \le 2^{n-1}$. For almost all boolean functions $f : \bitset^n \rightarrow \bitset$ with fixed weight  $\wt(f) = m(n)$, 
    \begin{equation}
    \label{equ:main_thm_lb}
    \avgq(f) \ge \min_{x \in \bitset^n} \cert_x(f) \ge \log \frac{\wt(f)}{\log n} - O\left( \log  \log \frac{\wt(f)}{\log n}\right),
    \end{equation}
    where $ \cert_x(f) $ denotes the size of the smallest certificate on input $x$. 
\end{theorem}

\begin{remark}
For boolean functions with $\wt(f) \ge 2^{n-1}$, we can obtain a similar bound by replacing $f$ with $\neg f$.
\end{remark}

Beyond fixed-weight functions, we also examine CNFs, circuits, and formulas, studying the connection between $\avgq(f)$ and criticality.

Rossman introduced the notion of \emph{criticality}, defined as the minimum value $\lambda \ge 1$ such that the following property holds: $\Pr_{\rho \sim \mathcal R_p} [\detq(f|_\rho) \ge t] \le (p\lambda)^t$ for any $p \in [0,1]$ and $t \in \mathbb N$. 
In terms of criticality, H{\aa}stad's switching lemma says every width-$w$ CNF is $O(w)$-critical \cite{Has86}; Rossman's switching lemma says every size-$s$ CNF is $O(\log s)$-critical \cite{Ros17, Ros19}; Rossman proved depth-$d$ size-$s$ $\ac$ circuits are $O(\log s)^{d - 1}$-critical \cite{Ros19}; Harsha et al. proved depth-$d$ size-$s$ $\ac$ formulas are $O(\frac{1}{d} \log s)^{d - 1}$-critical \cite{HMS23}.

For any $\lambda$-critical function $f$, applying a $(\frac{1}{2\lambda})$-random restriction and then querying the resulting subfunction via its optimal decision tree yields $\avgq(f) \leq n\left(1 - \frac{1}{\lambda}\right) + O\left(\sqrt{\frac{n}{\lambda}}\right)$ (Lemma \ref{lem:cir_general_ub} from Section \ref{ch:circuits}). 
Hence, criticality bounds imply average-case query complexity bounds for CNFs, formulas, and circuits.

For CNFs, circuits, or formulas, it is meaningful to understand whether the upper bounds on $\avgq(f)$ are tight or not. 
For example, consider a $w$-CNF $f$.
By Lemma \ref{lem:cir_general_ub} from Section \ref{ch:circuits}, we have $\avgq(f) \le n \left(1 - \frac{1}{O(w)}\right)$. If the bound were indeed tight, it would suggest that the $p$-random restriction, with $p = \frac{1}{10w}$, is essentially an optimal query algorithm for generic $w$-CNFs. 
Otherwise, either a better query algorithm exists, or a stronger version of the switching lemma can be established. Either way, the answer would be interesting.

Along this line, we show that there exists a DNF formula of width $w$ and size $\lceil2^w / w\rceil$ with $\avgq(f) = n(1 - \frac{\log n}{\Theta(w)})$. 
It indicates that even if there is a better query algorithm, the room for improvement is limited when $w$ is large. 

\begin{theorem} \label{thm:intro_circuit_lb}
For any integer $w \in [2\log n, n]$, there exists a boolean function $f : \bitset^n \rightarrow \bitset$ computable by a DNF formula of width $w$ and size $ \lceil2^w / w\rceil $ such that 
    \[
        \avgq(f) = n \left( 1 - \frac{\log n}{\Theta(w)} \right).
    \] 
\end{theorem}

Lastly, we define \emph{penalty shoot-out functions} in Appendix \ref{apdx:pso}, which are monotone balanced functions, such that the gap between $\detq(f)$ and $\avgq(f)$ is arbitrarily large.
Moreover, unlike the worst-case measures $\detq(f)$, $\randq(f)$, $\cert(f)$, $\bs(f)$, $\mathrm{s}(f)$, which are known to be polynomially related \cite{BD02, ABK16, Hua19, ABK+21}, no such polynomial relation holds between \emph{any} two of the average-case analogues\footnote{These average-case counterparts are defined in the uniform distribution. }
$\detq_{\mathrm{ave}}(f)$, $\randq_{\mathrm{ave}}(f)$, $\cert_{\mathrm{ave}}(f)$, $\bs_{\mathrm{ave}}(f)$, $\mathrm{s}_{\mathrm{ave}}(f)$, even for monotone balanced functions\footnote{Super-polynomial gaps can be demonstrated using the threshold function \cite{AD01}, the tribes function, and $\maj \circ \bfand$, all of which are monotone. An extra trick can make them balanced: given a monotone $f$, let $g = \maj(f,f^\dagger,z)$ for $z\in \bitset$, where $f^\dagger$ denotes $f$'s dual \cite{ODo14}. }.

\section{Preliminaries} \label{ch:prelim}

\subsection{Boolean functions}

Let $f : \bitset^n \rightarrow \bitset$ be a boolean function. The \textit{weight}, denoted by $\wt(f)$, is the number of inputs on which $f$ outputs 1.  Let $\bnm = \{f : \bitset^n\rightarrow \bitset \mid \wt(f) = m\}$ denote the set of all $n$-variable boolean functions with weight $m$. 

A \textit{restriction} $\rho : \{1, \ldots, n\} \rightarrow \{0,1,\star\}$ is a mapping fixing some variables to 0 or 1. We write $f|_\rho$ for the subfunction of $f$ obtained by fixing its input by $\rho$. Let \(\supp(\rho) = \rho^{-1}(\{0,1\})\) denote the support of the restriction $\rho$.
The weight of $x$, denoted by $|x|$, is the number of 1's in $x$. The bitwise negation of $x$ is denoted by $\neg x = (1 - x_1, \ldots, 1 - x_n)$. 

Let $F : \bitset^n \rightarrow \bitset$ and $G : \bitset^m \rightarrow \bitset$ be two boolean functions. 
Define the \textit{composition} $F \circ G$ by 
\[
    (F \circ G) (x) = F(G(x^{(1)}), \ldots, G(x^{(n)}))
\]
for $x = (x^{(1)}, \ldots, x^{(n)}) \in \bitset^{nm}$ and $x^{(i)} \in \bitset^m$. 

Define $x \preceq y$ if and only if $x_i \leq y_i$ for all $i \in \{1, 2, \ldots, n\}$. Say $f$ is \textit{monotone} if and only if $f(x) \leq f(y)$ for all inputs $x \preceq y$.

Given a boolean function $f:\{0, 1\}^n \to \{0, 1\}$, a certificate on input $x$ is a subset $S \subseteq \{1, \ldots, n\}$ such that $f(x) = f(y)$ for any input $y \in \{0, 1\}^n$ satisfying $x_i = y_i$ for all $i \in S$. The certificate complexity $\cert_x(f)$ on input $x$ is the size of the smallest certificate on input $x$. 

\subsection{Decision trees}

A (deterministic) decision tree $T$ is a binary tree. Each internal node is labeled by some integer $i \in \{1,2, \ldots, n\}$, and the edges and the leaves are labeled by 0 or 1. Repeatedly querying $x_i$ and following the edge labeled by $x_i$, the decision tree $T$ finally reaches a leaf and outputs the leaf's label, called the value $T(x)$ of $T$ on input $x$.  
The cost of deciding the value $T(x)$, denoted by $\cost(T, x)$, is the length of the root-to-leaf path which $T$ passes through. 
The depth of $T$ is the maximum cost $\max_{x \in \bitset^n} \cost(T,x)$. 
We say $T$ computes $f$ (with zero error) if $T(x) = f(x)$ for every $x \in \bitset^n$.
A query algorithm queries some variables and determines the value of the function; a query algorithm can be viewed as a family of decision trees. 

\subsection{Circuits and formulas}

A \emph{clause} is a logical OR of variables or their negations, and a \emph{term} is a logical AND of variables or their negations.
A \textit{conjunctive normal form (CNF)} formula is a logical AND of clauses, and a \textit{disjunctive normal form (DNF)} formula is a logical OR of terms. The \textit{size} of a CNF (respectively, DNF) formula is the number of the clauses (respectively, the terms).  The \textit{width} of a CNF (respectively, DNF) formula is the maximum variable number of the clauses (respectively, the terms). 

A \textit{circuit} $F$ is a directed acyclic graph with $n$ nodes of no incoming edge, called \textit{sources}, and a node of no outgoing edge, called \textit{sink}. Apart from the sources, the other nodes are called \textit{gates}. 
Each gate is labeled by AND, OR or NOT, and each AND (respectively, OR, NOT) node computes the logical AND (respectively, OR, NOT) of the values of its incoming nodes. The \textit{fan-in} of a gate is the number of its incoming edges, and the \textit{fan-out} of a gate is the number of its outgoing edges. The fan-in of NOT gates is fixed to $1$. The \textit{size} of $F$ is the number of the gates. The \textit{depth} is the length of the longest path between the sink and the sources. Obviously, every CNF or DNF formula can be represented as a circuit. An $\ac$ circuit is a circuit of polynomial size, constant depth and AND/OR gates with unbounded fan-in. A \emph{formula} is a circuit of gates with fan-out $1$.

\subsection{Query complexities}

Let $f:\{0,1\}^n \to \{0, 1\}$ be a boolean function.
The worst-case deterministic query complexity of a boolean function $f$, denoted by $\detq(f)$, is the minimum depth of decision trees that compute $f$. The average-case deterministic query complexity of a boolean function $f$, denoted by $\avgq(f)$, is the minimum average depth of zero-error deterministic decision trees that compute $f$ under a uniform input distribution.

\begin{definition}
The average-case deterministic query complexity of $f : \bitset^n \rightarrow \bitset$ under a uniform distribution is defined by 
\[
\avgq(f) = \min_{T} \mathop{\mathbb E}_{x \in \bitset^n} [\cost(T, x)],
\]
where $T$ is taken over all zero-error deterministic decision trees that compute $f$.
\end{definition}

$\avgq(f)$ turns out to equal the average-case zero-error \emph{randomized} query complexity, defined as  \(\min_{\mathcal{T}} \mathop{\mathbb E}_{x \in \bitset^n} [\cost(\mathcal{T}, x)],\) where $\mathcal{T}$ is taken over all zero-error \emph{randomized} decision trees that compute $f$ (see Remark 8.63 in \cite{ODo14}).

\section{Fixed-weight functions}

\subsection{Upper bound}

As a warm-up, the following proposition gives the exact value of $\avgq(f)$ for boolean functions $f$ with weight $1$, such as the AND function. 
For convenience, we say input $x$ is a \textit{black point} (with respect to $f$) if $f(x) = 1$.

\begin{proposition} \label{prop:dave-wt-1}
    $\avgq(f) = 2(1 - \frac{1}{2^n})$ for any $n$-variable boolean function $f$ with $\wt(f) = 1$.
\end{proposition}

\begin{proof}
    Let $f : \bitset^n \rightarrow \bitset$ be a boolean function with a unique black point $z \in \bitset^n$. 
    We prove $\avgq(f) = 2(1 - \frac{1}{2^n})$ by induction on $n$. 
    When $n = 1$, we have $\avgq(f) = 2(1 - \frac{1}{2^n}) = 1$. 
    
    Suppose an optimal query algorithm queries $x_i$ first. If $x_i \not= z_i$, the algorithm outputs $0$. Otherwise, the algorithm continues on the subfunction $f|_{x_i = z_i}$. Therefore,
    \begin{align*}
    \avgq(f)
    &= 1 + \Pr_{x \in \bitset^n} \left[ x_i \neq z_i\right] \cdot 0 + \Pr_{x \in \bitset^n} \left[x_i = z_i \right] \cdot \avgq(f|_{x_i = z_i}) \\
    &= 1 + \frac{1}{2} \cdot 2\left(1 - \frac{1}{2^{n-1}}\right) \\
    &= 2 \left(1 - \frac{1}{2^n}\right),
    \end{align*}
    where the second step is by the induction hypothesis.
\end{proof}

Next, we show a simple bound  $\avgq(f) \leq \log \wt(f) + O(1)$ for any $f$. 
Say a query algorithm is \textit{reasonable} if it terminates as soon as the subfunction becomes constant. 

\begin{lemma} \label{lem:ub_naive}
    $\avgq(f) \le \log \wt(f) + 2$ for any non-zero boolean function $f$.
\end{lemma}

\begin{proof}
    Let $m = \wt(f)$. We prove by induction on $m$ and $n$. When $m = 1$, we have $\avgq(f) = 2(1-\frac{1}{2^n}) \leq 2$ by Proposition \ref{prop:dave-wt-1}. When $n = 1$, $\avgq(f) \le 1$. 
    
    Suppose $x_i$ is queried first. Let $m_0 = \wt(f|_{x_i = 0})$ and $m_1 = \wt(f|_{x_i = 1})$.
    If $m_b = 0$ for some $b \in \bitset$, a reasonable algorithm will stop on a constant subfunction $f|_{x_i = b}$. Thus, 
    \[
    \avgq(f) \le 1 + \frac{1}{2}(\log m + 2) \leq \log m + 2
    \]
    by the induction hypothesis. Otherwise, by the induction hypothesis and the AM-GM inequality, we have
    \begin{align*}
        \avgq(f) & \le 1 + \frac{1}{2}(\log m_0 + 2) + \frac{1}{2}  (\log m_1 + 2) \\
        & = \log (2\sqrt{m_0 m_1}) + 2 \\
        & \leq  \log m + 2.  
    \end{align*}
\end{proof}

We introduce concepts that will be used later.
Suppose $f:\{0,1\}^n \to \{0, 1\}$ has $m$ black points $x^{(1)}, \ldots, x^{(m)} \in \bitset^n$ in lexicographical order.
We call $c_i = (x_i^{(1)}, \ldots, x_i^{(m)})$ the \textit{column pattern} of coordinate $i$. Coordinates $i, j$ are \emph{positively (negatively) correlated} if $c_i = c_j$ ($c_i = \neg c_j$). 
An \emph{equivalent coordinate set} (ECS) is a set of correlated coordinates.

Say a coordinate set $S \subseteq \{1, \ldots, n\}$ is \emph{pure} if each $c_i$ for $i \in S$ is either all-zero or all-one; otherwise, $S$ is \emph{mixed}.
For example, in Table \ref{tb:black_points}, the set $\{5, 9, {11}\}$ is pure, since $c_5 = c_9 = (0, 0)$ and $c_{11} = (1,1)$; the set $\{1, 2, 3\}$ is mixed, since $c_1 = c_3 = (0, 1)$ and $c_2 = (1, 0)$.

\begin{table}[!htbp]
    \centering
    \begin{tabular}{| c || c c c c c c c c c c c|} 
         \hline
         black points & $x_1$ & $x_2$ & $x_3$ & $x_4$ &  $x_5$ & $x_6$ & $x_7$ & $x_8$ & $x_9$ & $x_{10}$ & $x_{11}$ \\
         \hline\hline
         $x^{(1)}$  & 0 & 1 & 0 & 1 & 0 & 1 & 0 & 0 & 0 & 0 & 1
         \\
         \hline
         $x^{(2)}$  & 1 & 0 & 1 & 0 & 0 & 1 & 0 & 1 & 0 & 1 & 1
         \\
         \hline
    \end{tabular}
    \caption{A 11-variable boolean function with weight $2$.}
    \label{tb:black_points}
\end{table}

\begin{proposition} \label{prop:ecs_fact} 
    If coordinates $i, j$ are positively (negatively) correlated, then for any $x \in \{0,1\}^n$ with $f(x) = 1$, we have $x_i = x_j$ ($x_i \ne x_j$).
\end{proposition}

\begin{proposition} \label{prop:mixed_ecs} 
    Let $S$ be a mixed ECS. For any coordinate $i \in S$ and $v \in \bitset$, we have $\wt(f|_{x_i = v}) < \wt(f)$. 
\end{proposition}

\begin{proposition} \label{prop:id_exist}
Let $f : \bitset^n \rightarrow \bitset$ be a boolean function such that $n > k \cdot 2^{\wt(f)-1}$, then $f$ has an ECS of size at least $k + 1$.
\end{proposition}

It is straightforward to prove Propositions \ref{prop:ecs_fact} and \ref{prop:mixed_ecs}. 
Proposition \ref{prop:id_exist} follows from the pigeonhole principle, since there are $2^{m-1}$ distinct equivalence classes with respect to correlation. Using these facts, one can prove that any boolean function of weight $O(\log n)$ has constant $\avgq(f)$. 

\begin{lemma} \label{lem:ub_basic}
    Let $f : \bitset^n \rightarrow \bitset$, where $\wt(f) < \log n$. We have $\avgq(f) \le 5$. 
\end{lemma}

\begin{proof}
    Let $m = \wt(f) \ge 3$. (Lemma \ref{lem:ub_basic} follows directly from Lemma \ref{lem:ub_naive} if $\wt(f) < 3$.) We prove $\avgq(f) \leq 5$ by induction on $n$. 
    
    By Proposition \ref{prop:id_exist}, there exists a maximal ECS $I = \{i_1, \ldots, i_k\}$ of size $k \geq 3$, since $n > 2^{\wt(f)} = 2 \cdot 2^{\wt(f) - 1}$. Without loss of generality, assume that coordinates $i_1, \ldots, i_k$ are positively correlated. By Proposition \ref{prop:ecs_fact}, we have $x_{i_1} = \cdots = x_{i_k}$ for any black point $x \in \{0, 1\}^n$.
    
    If $|I| = n$, then any black point $x$ must satisfy $x_1 = \cdots = x_n$. Therefore, the only possible black points are the all-zero vector and the all-one vector, so there are at most 2 black points. By Lemma \ref{lem:ub_naive}, $\avgq(f) \le \log \wt(f) + 2 \le 3$.
    
    From now on, we assume $|I| < n$, and thus there exists a coordinate $j \not\in I$. Let $J$ be a maximal ECS that contains $j$. Notice that $I$ or $J$ is mixed, since at most one of $f$'s maximal ECSs can be pure.
    
    For notational convenience, let $\rho_{u, v}$ denote the restriction fixing $x_{i_1}$ and $x_{i_2}$ to $u$, $x_j$ to $v$, and leaving all other variables free.
    Our query algorithm $T$ is defined as follows:
    \begin{itemize}
        \item[(1)] query $x_{i_1}, x_{i_2}$;
        \item[(2)] output $0$ if $x_{i_1} \neq x_{i_2}$;
        \item[(3)] if $x_{i_1} = x_{i_2}$, then query $x_j$, and apply the query algorithm recursively on the subfunction. 
    \end{itemize}
    The query algorithm $T$ correctly computes $f$. Input $x$ cannot be a black point if $x_{i_1} \neq x_{i_2}$, since $i_1$ and $i_2$ are positively correlated (Proposition \ref{prop:ecs_fact}).
    
    For any $u, v \in \bitset$, the number of inputs of $f|_{\rho_{u,v}}$ is $n - 3$, and $\wt(f|_{\rho_{u,v}}) 
    \le m-1$ since $I$ or $J$ is mixed (Proposition \ref{prop:mixed_ecs}). 
    Observe that
    $n-3 > 2^m - 3 > 2^{m-1} \ge 2^{\wt(f|_{\rho_{u,v}})}$ for any $m \ge 3$.
    By the induction hypothesis, we have $\avgq(f|_{\rho_{u,v}}) \le 5$.
    
    Let us analyze the average cost of $T$.
    First, notice that the probability of querying exactly 2 variables is $\frac{1}{2}$; the probability of querying at least 3 variables is $\frac{1}{2}$. Thus, we conclude that 
    \begin{equation*}
        \avgq(f) \le 
        \mathop{\mathbb E}_{x} \left[ \cost(T,x) \right] 
        \le \frac{1}{2} \cdot 2 + \Pr_{x}\left[\cost(T, x) \ge 3 \right] \cdot  (3 + 5) = 5.
    \end{equation*}
\end{proof}

\begin{corollary} \label{cor:ub_small_weight}
    Let $f : \bitset^n \rightarrow \bitset$ be a boolean function. If $\wt(f) \le 4 \log n$, then $\avgq(f) \le 40$. 
\end{corollary}

\begin{proof}
    We have $\avgq(f_1(x) \vee \cdots \vee f_k(x)) \le \avgq(f_1) + \cdots + \avgq(f_k)$ for any $f_1, f_2, \ldots, f_k$. This is because we can query $f_1(x)$, $f_2(x)$, $\ldots$, $f_k(x)$ one by one and compute $f_1(x) \vee \cdots \vee f_k(x)$ afterward. The expected number of variables queried is at most the sum of the individual expectations.
    
    Let $k=8$ and $B_f$ denote $f$'s on-set (the set of inputs on which $f$ outputs $1$). 
    Partition $B_f$ into $k$ disjoint sets $B_{1} \ldots, B_{k}$, where $|B_i| \le \lceil \frac{4 \log n}{8} \rceil$. Each $B_i$ is the on-set of some function $f_i$. It can be verified that $f(x) = f_1(x) \vee \cdots \vee f_k(x)$ and that $\wt(f_i) = |B_{i}| \le \lceil \frac{4 \log n}{8} \rceil$.
    Note that $\wt(f_i) \le \lceil \frac{4 \log n}{8} \rceil < \frac{1}{2} \log n +1 \le \log n$ when $n \ge 4$. (When $n < 4$, the corollary holds clearly.) Thus, by Lemma \ref{lem:ub_basic}, $\avgq(f_i) \le 5$ for any $i$, implying that $\avgq(f) \le \sum_{i=1}^8 \avgq(f_i) \le 40$. 
\end{proof}

To prove Theorem \ref{thm:intro_bnm_ub}, we design a query algorithm and analyze its cost. Recall that in the proof of Lemma \ref{lem:ub_naive}, we considered \emph{any} reasonable query algorithm, which queries an arbitrary bit and terminates until the remaining function becomes constant.
Similarly, to prove Theorem \ref{thm:intro_bnm_ub}, we design a more sophisticated algorithm, which queries an arbitrary variable until the subfunction satisfies the following border condition: $\avgq(f) = O(1)$ if $\wt(f) \le  4\log n$ (Corollary \ref{cor:ub_small_weight}); then the query algorithm used in the proof of Corollary \ref{cor:ub_small_weight} is invoked. 

\begin{proof}[Proof of Theorem \ref{thm:intro_bnm_ub}]
    Let $m = \wt(f)$. If $m \le 4 \log n$, we have $\avgq(f) = O(1)$ by Corollary \ref{cor:ub_small_weight}. We will prove
    \begin{equation}
    \label{equ:ub_refined_goal}
    \avgq(f) \le \log \frac{m}{\log n} + \log \log \frac{m}{\log n} + 87,
    \end{equation}
    by induction on $n$ for any $f:\{0, 1\}^n \to \{0, 1\}$ with $\wt(f) = m \ge 4 \log n$.
    
    First, when $m \ge n$, the inequality \eqref{equ:ub_refined_goal} directly follows from Lemma \ref{lem:ub_naive} because $\log m + 2 \leq \log \frac{m}{\log n} + \log  \log \frac{m}{\log n} + 87$ if and only if $m \ge n^{2^{-85}} \log n$. 
    From now on, we assume that $m \le n$.
    
    Our query algorithm is as follows:
    \begin{itemize}
    \item[(1)] If $\wt(f) \ge n$, apply Lemma \ref{lem:ub_naive}, which implies \eqref{equ:ub_refined_goal} as we have shown.
    \item[(2)] Otherwise, query any 
    \[
    \ell  = \left\lceil \log \frac{m}{\log n} + \log \log \frac{m}{\log n} + 3 \right\rceil 
    \]
    variables, say,  $x_{i_1}, \ldots, x_{i_\ell}$.
    \item[(3)] Given the values of $x_{i_1}, x_{i_2}, \ldots, x_{i_\ell}$, say $x_{i_1} = c_1, \ldots, x_{i_\ell} = c_\ell$, apply our algorithm \emph{recursively} to the subfunction $f|_{x_{i_1} = c_1, \ldots, x_{i_\ell} = c_\ell}$. 
    \end{itemize}
        
    Let $\rho_1, \ldots, \rho_{2^\ell}$ enumerate all restrictions that fix $x_{i_1}, \ldots, x_{i_\ell}$, while leaving all remaining variables undetermined. 
    Averaging over all $\avgq(f|_{\rho_i})$, we have
    \begin{align}
         \avgq(f) \nonumber 
         \le \ell &+ \mathop{\mathbb E}_{i} [\avgq(f|_{\rho_i})] \nonumber  \\ 
        \le \ell &+ \Pr_{i}\left[\wt(f|_{\rho_i}) \le 4\log n\right] \cdot \mathop{\mathbb E}_{i}\left[ \avgq(f|_{\rho_i}) \mid \wt(f|_{\rho_i}) \le 4\log n \right] \nonumber  \\
        &+ \Pr_{i}\left[\wt(f|_{\rho_i}) > 4\log n\right] \cdot \mathop{\mathbb E}_{i}\left[ \avgq(f|_{\rho_i}) \mid \wt(f|_{\rho_i}) > 4\log n \right].  \label{eq:thm1-1}
    \end{align}
    We have $\wt(f) = \sum_{i = 1}^{2^\ell} \wt(f|_{\rho_i})$ and $\mathop{\mathbb E}_{i} [\wt(f|_{\rho_i})] = \frac{m}{2^\ell} \le \frac{1}{8} \cdot \frac{\log n}{\log \frac{m}{\log n}}$. 
    By Markov's inequality, 
    \begin{align*}
    \Pr_{i} \left[\wt(f|_{\rho_i}) > 4  \log n \right]  \leq \frac{\mathop{\mathbb E}_{i} [\wt(f|_{\rho_i})]}{4 \log n} \le \frac{1}{32} \cdot \frac{1}{\log\frac{m}{\log n}}.
    \end{align*}
    We bound $\avgq(f|_{\rho_i})$ based on the weight of $f|_{\rho_i}$.
    
    \textbf{Case 1: }$\wt(f|_{\rho_i}) \le 4\log n$.
    Since $m \le n$, we have $\ell \le \log n + \log \log n + 4$. So the number of variables in $f|_{\rho_i}$ is $n - \ell \ge n - (\log n + \log \log n + 4) \ge \sqrt{n}$ (when $n$ is large enough). Notice that  $\wt(f|_{\rho_i})\le 4 \log n \le 8\log (n-\ell)$. Thus, by Corollary \ref{cor:ub_small_weight}, we have $\avgq(f|_{\rho_i}) \le 80$.
     
    \textbf{Case 2: }$\wt(f|_{\rho_i}) > 4\log n$.
    Note that $\wt(f|_{\rho_i}) \ge 4 \log n > 4\log (n - \ell)$.
    By the induction hypothesis, we have 
    \begin{align*}
        \avgq(f|_{\rho_i}) & \le  \log \frac{m}{\log (n-\ell)} + \log \log \frac{m}{\log (n-\ell)} + 87 \\
         & \le  2\log \frac{m}{\log (n - \ell)} + 87 \\
        & \le 2\log \frac{m}{\log n} + 89,
    \end{align*}
    where the second step is because $\log \log \frac{m}{\log (n-\ell)} \le \log \frac{m}{\log (n-\ell)}$, and the third step is because $n - \ell \ge n - (\log n + \log \log n + 4)\ge \sqrt{n}$.
    
    Combining the two cases and plugging them into (\ref{eq:thm1-1}), we have
    \begin{align*}
         &\  \avgq(f) \\
        \le &\   \log \frac{m}{\log n} + \log \log \frac{m}{\log n} + 4 + 80 + \frac{1}{32} \cdot \frac{1}{\log \frac{m}{\log n}} \cdot { \left(2\log \frac{m}{\log n} + 89 \right)}  \nonumber  \\
        = &\  \log \frac{m}{\log n} + \log \log \frac{m}{\log n} + 84 + \frac{1}{16} + \frac{89}{32 \cdot \log \frac{m}{\log n} } \\
         \le &\  \log \frac{m}{\log n} + \log \log \frac{m}{\log n} + 87.
    \end{align*}
    
    To conclude, if $m \ge 4\log n$, the right-hand side of \eqref{equ:ub_refined_goal} is at most $\log \frac{m}{\log n} + O(\log \log \frac{m}{\log n})$, completing the proof of Theorem \ref{thm:intro_bnm_ub}.
\end{proof}

\subsection{Lower bound}

In this section, we prove Theorem \ref{thm:intro_bnm_lb}, showing that Theorem \ref{thm:intro_bnm_ub} is tight up to an additive logarithmic term.

To illustrate the idea of the proof, let us take the $\bfxor_n$ function as an example. 
Regardless of which variable is queried next, the black points are evenly partitioned, and the subfunction's weight is exactly halved.
Since $\bfxor_n$ has weight $2^{n-1}$ and the algorithm must continue until the subfunction becomes constant, it must query all $n$ variables for every input. Thus, $\avgq(\bfxor_n) = n$.
Similarly, the key idea of our proof is to show that most boolean functions exhibit a similar property: regardless of which variable is queried next, the black points are split into two roughly equal halves.
In other words, for almost all $f \in \bnm$, where $\bnm = \{f : \bitset^n\rightarrow \bitset \mid \wt(f) = m\}$, we shall prove $\wt(f|_P)$ is ``close'' to $2^{-k}m$ for \emph{any} tree path $P$ querying $k = \epsilon \log m$ variables. (Ignoring the output of $P$, we view a tree path $P$ as a restriction, with $f|_P$ representing the subfunction restricted to $P$.)

Now, we explain the proof strategy in more detail. 
To sample $f \in \bnm$ uniformly, a straightforward approach proceeds as follows: (1) randomly select $m$ distinct inputs $x^{(1)}, \cdots, x^{(m)} \in \bitset^n$; (2) set $f(x^{(i)}) = 1$ for all $i$; and (3) set the remaining inputs to $0$.
This can also be done by repeatedly drawing $m$ vectors from $\{0, 1\}^n$ without replacement and placing them into $m$ vectors $y^{(1)}, \ldots, y^{(m)} \in \{0,1\}^n$. However, to estimate the probability that $\wt(f|_P)$, where $f \in \bnm$, is close to $2^{-k} m$ for any length-$k$ tree path, we adopt a different sampling method.

Fix a tree path $P$, viewed as a restriction. Instead of sampling a random $f \in \bnm$ and then estimating $\wt(f|_P)$, we choose to sample $\wt(f|_P)$ directly, where $f \in \bnm$, using the following method.

Fix a tree path $P$ of length $k$, where
\begin{equation}
\label{equ:path_P_def}
P = x_{i_1} \xrightarrow[]{v_1} x_{i_2} \xrightarrow[]{v_2} \cdots  \xrightarrow[]{} x_{i_k} \xrightarrow[]{v_k} c
\end{equation}
and $c \in \bitset$ is the output of the path.
We denote the restriction $f|_{x_1 = v_1, \cdots, x_k = v_k}$ by $f|_P$. In $k$ rounds, for $j = 1, \ldots, k$, we sample without replacement from a box with $2^{n-j}$ 0's and $2^{n-j}$ 1's; place the numbers in the $i_j$-th position of each vector, and discard vectors with $(\neg v_j)$ at $i_j$-th position. 
(We can safely discard these vectors, because they are not counted in the weight of $f|_P$.) 
At the end of $k$ rounds, $\wt(f|_P)$ vectors remain.

Specifically, we sample $\wt(f|_P)$ as follows, given a fixed path $P$ defined in \eqref{equ:path_P_def}, where $f \in \bnm$ uniformly:
\begin{enumerate}
    \item[(1)] Let $y^{(1)}, \ldots, y^{(m)} \in \{0,1,\star\}^n$ be the $m$ vectors, where all elements are set to $\star$ initially.
    \item[(2)] In the first round, we sample $t_0 = m$ numbers without replacement from a box with $2^{n-1}$ zeros and $2^{n-1}$ ones. We then assign these numbers sequentially to the positions $y_{i_1}^{(1)}, \ldots, y_{i_1}^{(m)}$, that is, the $i_1$-th position of all the $m$ vectors. After that, discard the vectors with $(\neg v_1)$ at position $i_1$, that is, $y_{i_1}^{(p)} = \neg v_1$.
    Let $t_1$ be number of remaining vectors, where $t_1$ is a random variable equal to $\wt(f|_{x_{i_1}=v_1})$.
    \item[(3)] In the second round, we sample $t_1$ numbers without replacement from a box with $2^{n-2}$ zeros and $2^{n-2}$ ones, since $f|_{x_{i_1}=v_1, x_{i_2}=0}$ and $f|_{x_{i_1}=v_1, x_{i_2}=1}$ have $2^{n-2}$ inputs. Assign these numbers sequentially to the $i_2$-th position of the remaining $t_1$ vectors, and discard the vectors with $(\neg v_2)$ at position $i_2$, i.e., $y_{i_2}^{(p)} = \neg v_2$.
    Let $t_2$ be number of remaining vectors, where $t_2$ is a random variable equal to $\wt(f|_{x_{i_1} = v_1, x_{i_2} = v_2})$.
    \item[(4)] Proceed for $k$ rounds. The number of remaining vectors $t_k$ is a random variable equal to $\wt(f|_P) = \wt(f|_{x_{i_1} = v_1, \ldots, x_{i_k} = v_k})$. 
\end{enumerate}

Recall that $t_j$ is the number of vectors remaining after the $j$-th round, and $t_k = \wt(f|_P)$. If path $P$ correctly computes $f$ (on input $x \in \{0,1\}^n$ such that $x_{i_1} = v_1$, $\ldots$, $x_{i_k} = v_k$), then we must have $\wt(f|_P) = 0$ or $\wt(f|_P) = 2^{n - k}$.
Intuitively, in each round, $t_i \approx \frac{1}{2}t_{i-1}$ holds with high probability by Hoeffding's inequality, so it takes $\Omega(\log n)$ rounds to make $t_k = 0$. Thus, it is unlikely that a ``short'' path computes $f$. 

\begin{definition} [$\delta$-parity path] \label{def:delta_path}
Let $P = x_{i_1} \xrightarrow[]{v_1} x_{i_2} \xrightarrow[]{v_2} \cdots  \xrightarrow[]{} x_{i_k} \xrightarrow[]{v_k} c$ be a path of length $k$. Let $\rho_j$ denote the restriction that fixes $x_{i_p}$ to $v_p$ for $p = 1, \ldots, j$, leaving other variables undetermined. The path $P$ is called \emph{$\delta$-parity} with respect to $f : \bitset^n \rightarrow \bitset$ if 
    \[
    \frac{1}{2}(1 - \delta) \le \frac{\wt(f|_{\rho_j})}{\wt(f|_{\rho_{j-1}})} \le \frac{1}{2}(1 + \delta)
    \]
    for each $j = 1, \ldots, k$. 
\end{definition}

\begin{lemma} [Hoeffding's inequality\cite{Hoe63, Ser74}] 
    \label{lem:hoeffding_ineq} Let $X_1, X_2, \ldots, X_m$ be independent random variables such that $0 \le X_i \le 1$, and let $S_m = X_1 + X_2 + \ldots + X_m$. For any $t > 0$, we have
    \begin{equation}
    \label{equ:hoef}
    \Pr[|S_m - \mathbb{E}[S_m]| \ge t] \le 2 \exp \left(-\frac{2t^2}{m}\right).
    \end{equation}
    The inequality \eqref{equ:hoef} also holds when $X_1, \ldots, X_m$ are obtained by sampling without replacement.
\end{lemma}

\begin{lemma} \label{lem:parity_path_2}
    Let $f:\{0, 1\}^n \to \{0,1\}$ be a boolean function with $\wt(f) = m$.
    Let $P$ be a decision tree path of length at most $\epsilon \log m$. For any $\delta \in (0, \frac{1}{2 \epsilon\log m}]$ and $\epsilon \in (0, 1)$, 
    \[
    \Pr_{f \sim \bnm} \left[ P \text{ is not } \delta\text{-parity for } f \right] < 2\epsilon\log m \cdot \exp\left( - \frac{1}{2} \cdot \delta^2 m^{1-\epsilon} \right).
    \]
\end{lemma} 

\begin{proof}
    Let $b \le \epsilon\log m$ denote the length of $P$. 
    The random variable $t_j$ equals $\wt(f|_{\rho_j})$, where $f$ is sampled uniformly at random from $\bnm$.
    Let $X_{j,1}$, $\ldots$, $X_{j, t_{j-1}} \in \{0,1\}$ be random variables indicating whether each of the $t_{j-1}$ vectors is in the on-set of $f|_{\rho_j}$, i.e., whether it remains after the $j$-th round.  Random variables $X_{j,1}, \ldots, X_{j, t_{j-1}} \in \{0,1\}$ are obtained from $2^{n-j}$ zeros and $2^{n-j}$ ones by sampling without replacement.
    We have $t_j = X_{j, 1} + \cdots + X_{j, t_{j-1}}$ and $\mathbb E[t_j] = \frac{1}{2}t_{j-1}$.
    
    Let  $\alpha = \frac{1}{2}(1-\delta)$ and $\beta = \frac{1}{2}(1+\delta)$. We have
    \begin{align}
         &\ \Pr_{ f \sim \bnm} \left[ P \text{ is not } \delta\text{-parity with respect to } f \right] \nonumber\\ 
        =&\  1 - \Pr \left[ \bigwedge_{j = 1}^{b} t_j \in \left[ \alpha t_{j-1}, \beta t_{j-1} \right] \right] \nonumber\\
        =&\ \Pr \left[ \exists 1 \leq j \leq b \,\text{ s.t. }\, t_j \not\in \left[ \alpha t_{j-1}, \beta t_{j-1} \right] \wedge \left( \bigwedge_{k = 1}^{j-1}  t_k \in \left[ \alpha t_{k-1}, \beta t_{k-1} \right]\right)\right]. \label{equ:delta_parity_prob}
    \end{align}
    Let $A_j$ be the event that $t_j \in \left[ \alpha t_{j-1}, \beta t_{j-1} \right]$, and let $B_j$ be the event $t_{j} \in \left[ \alpha^{j} m, \beta^{j} m \right]$. By a union bound, \eqref{equ:delta_parity_prob} is at most
    \begin{equation} \label{eq:lb_lem_3}
        \sum_{j = 1}^b \Pr \left[ \neg A_j \wedge \left( \bigwedge_{k = 1}^{j-1} A_k \right)\right] 
        \leq \sum_{j = 1}^b \Pr \left[ \neg A_j \wedge B_{j-1} \right]. 
    \end{equation}
    
    If event $A_j$ does not occur, we have $t_j > \beta t_{j-1}  = \frac{1}{2}(1 + \delta) t_{j-1}$ or $ t_j < \alpha t_{j-1} = \frac{1}{2}(1 - \delta) t_{j-1}$, which implies $|t_j - \mathbb E[t_j]| > \frac{1}{2}\delta t_{j-1}$.
    By Hoeffding’s inequality (Lemma \ref{lem:hoeffding_ineq}), 
    \begin{equation}
    \label{equ:pb_Bj_ub}
            \Pr \left[ \neg A_j \wedge B_{j-1} \right]  \leq 2 \exp \left( -\frac{1}{2} \delta^2 t_{j-1} \right) \leq 2 \exp \left( -\frac{1}{2} \alpha^{j-1} \delta^2 m \right),
    \end{equation}
    since $ t_{j-1} \geq \alpha^{j-1} m$ by $B_{j-1}$. 
    
    Finally, plugging \eqref{equ:pb_Bj_ub} into \eqref{eq:lb_lem_3}, we have
    \begin{align*}
        &\ \Pr_{ f \sim \bnm} \left[ P \text{ is not } \delta\text{-parity with respect to } f \right]& &\\ 
        \leq&\ \sum_{j = 1}^b 2 \exp \left( -\frac{1}{2} \cdot \alpha^{j-1} \delta^2 m \right) & & \\
         \leq&\  2b \cdot \exp \left( -\frac{1}{2} \cdot \alpha^{b-1} \delta^2 m \right) & & \left( \alpha = \frac{1}{2}(1 - \delta) < 1 \right) \\ 
         \leq&\  2b\cdot \exp \left( -\frac{1}{2}\cdot \frac{1}{2^{b-1}} \left(1 - \frac{1}{2\epsilon \log m}\right)^{\epsilon \log m - 1} \delta^2 m \right) & & \left(\delta \leq \frac{1}{2 \epsilon \log m} \right) \\
         \leq&\ 2b \cdot  \exp \left( - \frac{1}{2} \delta^2 \cdot \frac{m}{2^b} \right) & & \left(   \left(1 - \frac{1}{2x}\right)^{x-1}  > \frac{1}{2}  \right) \\
         \leq&\  2\epsilon\log m \cdot  \exp \left( -\frac{1}{2} \cdot   \delta^2 m^{1- \epsilon} \right). & & \Big( b \le \epsilon\log m \Big)
\end{align*}
\end{proof}

\begin{definition} [$(t, \delta)$-parity function] Let $f : \bitset^n \rightarrow \bitset$ be a boolean function. The function $f$ is called \emph{$(t,\delta)$-parity} if any path of length at most $t$ is a $\delta$-parity path with respect to $f$. 
\end{definition}

\begin{lemma}  \label{lem:parity_path_1}
    Let $f : \bitset^n \rightarrow \bitset$ be a $(t,\delta)$-parity function with $\wt(f) \le 2^{n-1}$ satisfying $1 \le t \le \log \wt(f) - 1$  and $\delta \le \frac{1}{2 t} $. Then, $\min_{x\in \bitset^n} \cert_x(f) \ge t$. 
\end{lemma}

\begin{proof}
    Let  $P = x_{i_1} \xrightarrow[]{v_1} x_{i_2} \xrightarrow[]{v_2} \cdots  \xrightarrow[]{} x_{i_k} \xrightarrow[]{v_k} c$ be any decision tree path of length $k \le t$. By Definition \ref{def:delta_path}, we have
        \begin{equation} \label{eq:parity_weight_ineq}
            \frac{1}{2^j}(1 - \delta)^j \le
            \frac{\wt(f|_{\rho_{j})}}{\wt(f)} = 
            \frac{\wt(f|_{\rho_1})}{\wt(f)} \cdot  \frac{\wt(f|_{\rho_2})}{\wt(f|_{\rho_1})} \cdots  \frac{\wt(f|_{\rho_{j}})}{\wt(f|_{\rho_{j-1}})}
            \le 
            \frac{1}{2^j}(1 + \delta)^j
        \end{equation}
        for $j = 1,2,\ldots, k$, where $\rho_j$ is the restriction that fixes $x_{i_p}$ to $v_p$ for $p = 1, 2, \ldots, j$.
        Thus, we have $\frac{1}{2^j}(1 - \delta)^j \wt(f) \le \wt(f|_{\rho_{j}}) \le \frac{1}{2^j}(1 + \delta)^j \wt(f)$. 
    Note that \eqref{eq:parity_weight_ineq} holds for any decision tree path of length $j$. So we have 
    \begin{equation}
    \label{equ:wtfr_lb_up}
    \frac{1}{2^j}(1 - \delta)^j \wt(f) \le \wt(f|_{\rho}) \le \frac{1}{2^j}(1 + \delta)^j \wt(f)
    \end{equation}
    for any restriction $\rho$ fixing $j$ variables, where $j \le t$.
        
    On one hand, from \eqref{equ:wtfr_lb_up}, we have
    \begin{align*}
    \wt(f|_\rho) 
    &\le (1 + \delta)^t 2^{n - j - 1} & & \Big( \wt(f) \le 2^{n-1} \Big) \\
    &\le \frac{1}{2}  \left(1 + \frac{1}{2t}\right)^t 2^{n - j} & & \left(  \delta \le \frac{1}{2t} \right) \\
    &\le \frac{\sqrt{e}}{2} \cdot 2^{n - j} \le 0.8244 \cdot 2^{n - j}. & &  \left( \left(1 + \frac{1}{2n}\right)^n \le \sqrt{e} \right)
    \end{align*}
    On the other hand, by \eqref{equ:wtfr_lb_up}, we have
    \begin{align*}
    \wt(f|_\rho) 
    &\ge  (1 - \delta)^t 2^{t - j + 1}   & & \Big(t \le \log \wt(f) - 1 \Big) \\
    &\ge 2\left(1 - \frac{1}{2t}\right)^t   & & \left(\delta \le \frac{1}{2t} \text{ and } j \le t \right) \\
    &\ge 1. & & \left( \left(1 - \frac{1}{2x}\right)^x \ge \frac{1}{2} \text{ when } x \ge 1 \right)
    \end{align*}
    Thus, $\wt(f|_\rho)$ is less than $2^{n - j}$ and  larger than $0$, which implies $f|_\rho$ cannot be constant for any restriction $\rho$ fixing at most $t$ variables. Therefore, we conclude $\min_{x\in \bitset^n} \cert_x(f) \ge t$.
\end{proof}

Finally, one can prove Theorem \ref{thm:intro_bnm_lb} using Lemmas \ref{lem:parity_path_2} and \ref{lem:parity_path_1}.

\begin{proof} [Proof of Theorem \ref{thm:intro_bnm_lb}]
    Let  $m = m(n)$ and
    \[
    \epsilon = 1 - \frac{1}{\log m} \left(\log \log n + 3\log \log \frac{m}{\log n}+ 5\right).
    \]
    Since $\cost(T,x) \ge \cert_x(f)$ for any $x \in \bitset^n$, $\avgq(f) \ge \min_{x \in \bitset^n} \cert_x(f)$, 
    Our goal is to prove
    \[
    \Pr_{{f} \sim \bnm} \left[\min_{x \in \bitset^n} \cert_x(f) < \epsilon \log m = \log \frac{m}{\log n} - 3 \log \log \frac{m}{\log n} - 5 \right] \to 0
    \]
    as $n \rightarrow \infty$. Since $4\log n \le m \le 2^{n-1}$, $m,n$ tend to infinity simultaneously.
    
    Let $t = \epsilon\log m$ and $\delta = (2\log \frac{m}{\log n})^{-1} \le (2\epsilon \log m)^{-1}  = \frac{1}{2t}$. 
    Let $\mathrm{len}(P)$ denote the length of a tree path $P$.
    By Lemma \ref{lem:parity_path_1}, 
    if  $\min_{x \in \bitset^n} \cert_x(f)  < t$, then $f$ is not $(b,\delta)$-parity. That is, there exists a path $P$ being not $\delta$-parity. 
    Thus,
    \begin{align}
        &\  \Pr_{{f} \sim \bnm} \left[\min_{x \in \bitset^n} \cert_x(f) < \epsilon \log m\right] \nonumber \\
        \le&\ \Pr_{{f} \sim \bnm} \left[ \exists \text{ tree path } P \,\text{ with }\, \mathrm{len}(P) < t \text{ such that } P \text{ is not } \delta\text{-parity}\right]. \label{equ:lb_goal}
    \end{align}
    By a union bound, \eqref{equ:lb_goal} is at most
    \begin{align}
        &\ \sum_{\mathrm{len}(P) <  t} \Pr_{ f \sim \bnm} \left[P \text{ is not } \delta\text{-parity for } f \right] &  \nonumber \\
        \leq&\ \sum_{k = 0}^{\epsilon \log m - 1} {n \choose k} k!2^k \cdot  \exp\left( - \frac{1}{2} \delta^2 m^{1-\epsilon} \right)  & &  (\text{Lemma \ref{lem:parity_path_2}}) \nonumber \\
        \leq&\ n^{\epsilon \log m} (\epsilon \log m)^{2\epsilon \log m} \cdot (2\epsilon\log m) \cdot\exp\left( - \frac{1}{2} \delta^2 m^{1-\epsilon} \right) & &  \nonumber \\
        \leq&\ \exp\left(\ln n \cdot 4\epsilon \log m\right) \cdot  \exp\left( - \frac{m^{1- \epsilon}}{8 (\log \frac{m}{\log n})^2} \right). &  \label{eq:bnm_lb_prob}
    \end{align}
    Then, since $\epsilon \log m \le \log \frac{m}{\log n}$, \eqref{eq:bnm_lb_prob} is at most
    \begin{align*}
        &\ \exp\left( \log n \cdot (4\ln 2) \cdot \log \frac{m}{\log n} - \frac{m^{1- \epsilon}}{8(\log \frac{m}{\log n})^2} \right) & \\
        =&\ \exp\left( - 4(1 - \ln 2)\log \frac{m}{\log n} \cdot \log n \right) &  \\
        \leq&\ \exp\left( -  8 (1 - \ln 2) \log n \right), & \left( m \ge 4 \log n \right)
    \end{align*}
    which tends to zero as $n \rightarrow \infty$.
    Therefore, we conclude
    \[
        \avgq(f) \ge \min_{x \in \bitset^n} \cert_x(f) \geq \epsilon \log m = \log \frac{m}{\log n} - 3 \log \log \frac{m}{\log n} - 5
    \]
    for almost all functions $f\in \bnm$. That is, $\avgq(f)$ is at least $ \log \frac{m}{\log n} - O(\log \log \frac{m}{\log n})$, since $m \ge 4\log n$.
\end{proof}

\section{DNFs, circuits, and formulas} \label{ch:circuits}

In this section, we study $\avgq(f)$ of circuits that consist of AND, OR, NOT gates with unbounded fan-in.

As a warm-up, we show $\avgq(F) = O(s)$ for general size-$s$ circuits $F$. The bound is tight up to a multiplicative factor, since $\avgq(\bfxor_n) = n$, and $\bfxor_n$ is computable by a circuit of size $O(n)$ and depth $O(\log n)$.

\begin{proposition} \label{lem:cir_general_ub}
    $\avgq(F) \leq 2s$ for every circuit $F$ of size $s$.
\end{proposition}

\begin{proof}
Notice that the average cost of evaluating each AND/OR/NOT gate does not exceed $2$ (Proposition \ref{prop:dave-wt-1}). Therefore, it takes at most $2s$ queries on average to evaluate $s$ gates.
\end{proof}

\begin{definition} 
    A \emph{$p$-random restriction}, denoted by $\prandrtrc$, is a distribution over restrictions leaving $x_i$ unset with probability $p$ and fixing $x_i$ to $0$ or $1$ with equal probability $\frac{1}{2}(1 - p)$ independently for each $i = 1,2,\ldots, n$. 
\end{definition}

\begin{definition} [\cite{Ros17, Ros19}]
    A boolean function $f$ is \emph{$\lambda$-critical} if 
    \[
        \Pr_{\rho \sim \mathcal{R}_p} \left[ \detq(f|_\rho) \geq t \right] \leq (p\lambda)^{t}
    \]
    for any $p \in [0,1]$ and $t \in \mathbb N$.
\end{definition}

The next lemma gives an upper bound on $\avgq(f)$ for $\lambda$-critical functions. 

\begin{lemma} \label{lem:cri_to_avgq} 
    Let $f : \bitset^n \rightarrow \bitset$ be $\lambda$-critical. Then \begin{equation} \label{eq:cri-to-avgq}
        \avgq(f) \le n \left(1 - \frac{1}{ \lambda}\right) + 2\sqrt{\frac{n}{\lambda}}.
    \end{equation}
\end{lemma}

\begin{proof}
    Let $\epsilon > 0$ and $p = \frac{1}{(1 + \epsilon) \lambda}$. Since $f$ is $\lambda$-critical, we have $\Pr_{\rho \sim \mathcal{R}_p}[\detq(f|_{\rho}) \geq t] \leq (1 + \epsilon)^{-t}$. 
    Consider a query algorithm that queries each variable independently with probability $1-p$, and then applies a worst-case optimal query algorithm to $f|_\rho$. We have
    \begin{align*}
        \avgq(f) 
        &\leq  \mathop{\mathbb E}_{\rho \sim \mathcal{R}_p} [|\supp(\rho)| + \detq(f|_\rho)]\\
        &=  n(1 - p) + \sum_{t = 1}^n\Pr_{\rho \sim \mathcal{R}_p} [\detq(f|_{\rho}) \ge t] \\
        &\leq  n \left( 1 - p \right) + \sum_{t = 0}^\infty \frac{1}{(1 + \epsilon)^t} \\
        & =  n \left( 1 - \frac{1}{(1+\epsilon)\lambda} \right) + \frac{1 + \epsilon}{\epsilon} \\
        & = n \left(1 - \frac{1}{\lambda}\right) + \frac{n}{\lambda} \cdot \frac{\epsilon}{1+ \epsilon} + \frac{1+\epsilon}{\epsilon}.
    \end{align*}
    Let $\alpha = \frac{n}{\lambda} \ge 1$. The function $h(\epsilon) = \frac{\alpha \epsilon}{1 + \epsilon} + \frac{1 + \epsilon}{\epsilon}$ attains its minimum at $\epsilon = \frac{1}{\sqrt{\alpha} - 1}$, where $h(\frac{1}{\sqrt{\alpha} - 1}) = 2\sqrt{\alpha}$. Thus,
    \[
        \avgq(f) \le n \left(1 - \frac{1}{\lambda}\right) + 2\sqrt{\frac{n}{\lambda}}. 
    \]
\end{proof}

\begin{remark}
Alternatively, one can prove $\avgq(f) \le n(1 - \frac{1}{2\lambda}) + O(1)$ by combining the OS inequality $\avgq(f) \leq \log \dtsize(f)$ \cite{OS06} and the bound $\dtsize(f) \leq O(2^{n(1-\frac{1}{2\lambda})})$ \cite{Ros19}. 
\end{remark}

By combining Lemma \ref{lem:cri_to_avgq} with the existing bounds on criticality for CNFs, bounded-depth circuits, and formulas \cite{Has86, Ros17, Ros19, Has14, HMS23}, the following upper bounds on $\avgq(f)$ can be derived.

\begin{corollary} [\cite{Has86}] \label{cor:crit-has86-avgq} Let $f: \bitset^n \rightarrow \bitset$ be computable by a CNF/DNF of width $w$. Then
    \[
        \avgq(f) \leq n\left(1 - \frac{1}{O(w)}\right).
    \] 
\end{corollary}

\begin{corollary} [\cite{Ros17, Ros19}] \label{cor:crit-ros17-avgq} Let $f: \bitset^n \rightarrow \bitset$ be computable by a CNF/DNF of size $s$. Then
    \[
        \avgq(f) \leq n\left(1 - \frac{1}{O(\log s)}\right).
    \] 
\end{corollary}

\begin{corollary} [\cite{Has14, Ros17}] \label{cor:crit-has14-avgq}  Let $f: \bitset^n \rightarrow \bitset$ be computable by a circuit of depth $d$ and size $s$. Then
\[
\avgq(f) \leq n \left(1 - \frac{1}{O(\log s)^{d-1}}\right).
\]
\end{corollary}

\begin{corollary} [\cite{HMS23}] \label{cor:crit-hms23-avgq}
Let $f: \bitset^n \rightarrow \bitset$ be computable by a formula of depth $d$ and size $s$. Then
\[
\avgq(f) \leq n \left(1 - \frac{1}{O(\frac{1}{d}\log s)^{d-1}}\right).
\]
\end{corollary}

It is natural to ask whether the upper bounds above are tight. A positive answer would suggest that random restrictions with the same probability $p$ (as was used in the proof of the aforementioned results) are optimal.
Toward this goal, we prove Theorem \ref{thm:intro_circuit_lb}, which says there exists a DNF of width $w$ and size $\lceil2^w / w\rceil$ such that $\avgq(f) = n (1- \frac{\log n}{\Theta(w)})$. 

Here, we provide an outline of the proof and briefly explain how to find such a DNF formula. 
In contrast to the $O(1)$ average cost to determine the output of the OR function  under a uniform input distribution (Proposition \ref{prop:dave-wt-1}), it costs $n(1 - o(1))$ on average under a $p$-biased input distribution when $p = o(1/n)$ (Exercise 8.65 in \cite{ODo14}). Our approach is to employ a biased function $g$ (given by Theorem \ref{thm:intro_bnm_lb}) with $p = \Pr_{x \in \bitset^n}[g(x) = 1]$ and $\avgq(g) = n(1 - o(1))$ as a ``simulator'' of $p$-biased variable. Then, we show the composition $\bfor \circ g$ is hard to query under a uniform distribution and is computable by a somewhat small DNF formula. As such, Theorem \ref{thm:intro_circuit_lb} follows.

\begin{proof} [Proof of Theorem \ref{thm:intro_circuit_lb}]
    Let $m = \lceil \frac{2^w}{2n} \rceil$ and $h = \lfloor \frac{n}{w} \rfloor$ and $s = \lceil 2^w/ w \rceil$. Observe that 
    \[
    mh \le \left(\frac{2^w }{2n} + 1\right)\cdot \frac{n}{w} = \frac{2^w}{2w} + \frac{n}{w} \le \left\lceil \frac{2^w}{w} \right\rceil = s.
    \]
    Since $n \ge w \ge 2\log n$, we have 
    \(
        m \le \frac{2^w }{2n} + 1  \le 2^{n-1}
    \)
    and
    \(
        m \ge \frac{2^w }{2n} \ge \frac{n^2}{2n}= \frac{n}{2} \ge 4\log w.
    \)
    By Theorem \ref{thm:intro_bnm_lb}, there exists $g \in \mathcal{B}_{w,m}$ such that
    \[
    d = \min_{y \in \bitset^w} \cert_y(g) \ge \log \frac{m}{\log w} - O \left(\log \log \frac{m}{\log w}\right) = w - \log n - O(\log w).
    \]
    Let $p = \frac{m}{2^w}$ denote the probability that  $g$ outputs $1$. Note that $p \le \frac{1}{2n} + \frac{1}{2^w} \le \frac{1}{n}$.
         
    Let
    \begin{equation}
    \label{equ:def_f_disjoint}
    f(x) = \bigvee_{k = 1}^{h} g(x^{(k)}),
    \end{equation}
    where $x = (x^{(1)}, \ldots , x^{(h)}) \in \bitset^n$ and $x^{(1)}, \ldots , x^{(h)} \in \bitset^w$. It is obvious that $f$ is computable by a DNF formula of width $w$ and size $mh \le s$, because each individual $g$ is computable by a DNF formula of width $w$ and size $m$. 
     
    Let $T$ be any query algorithm computing $f$. Let us condition on the event $g(x^{(1)}) = \cdots = g(x^{(h)}) = 0$, which happens with probability $(1-p)^h$. The algorithm $T$ needs to query at least $d$ variables to evaluate each clause $g(x^{(k)})$. 
    If $g(x^{(1)}) = \cdots = g(x^{(h)}) = 0$, then $T$ queries at least $h d$ variables. Thus,
    \begin{align}
    \avgq(f) \ge 
    &\ h d \cdot \Pr_{x \in \bitset^n}\left[g(x^{(1)}) = \cdots = g(x^{(h)}) = 0 \right] && \nonumber\\
    \ge&\ n \left(1 - \frac{\log n + O(\log w)}{w}\right) (1 - p)^h & & \Big(d \ge w - \log n - O(\log w)\Big) \nonumber\\ 
    \ge&\ n \left( 1 - \frac{\log n +O(\log w)}{w} \right) (1 - ph) && \Big(1-ph \le (1 -p)^h  \Big) \nonumber\\
    =&\ n \left(1 - \frac{\log n}{\Omega(w)}\right).  & & \left(ph \le \frac{1 }{w} \right)\label{equ:thm1.3-lb}
    \end{align}
    
    On the other hand, we can query $f$ by evaluating all the clauses one by one. By Lemma \ref{lem:ub_naive}, $\avgq(g) \le \log m + 2 \le w - \log n + 2$.  Thus, 
    \begin{equation}
        \avgq(f) \le h \cdot \avgq(g) \le n \left(1 - \frac{\log n - 2}{w}\right) = n \left(1 - \frac{\log n}{O(w)}\right). \label{equ:thm1.3-ub}
    \end{equation}
    Finally, combining \eqref{equ:thm1.3-lb} and \eqref{equ:thm1.3-ub}, we conclude $\avgq(f) = n \left(1 - \frac{\log n}{\Theta(w)}\right)$. 
\end{proof}

\section{Conclusion}

In this paper, we studied the average-case query complexity of boolean functions under the uniform distribution. We prove an upper bound on $\avgq(f)$ in terms of its weight; on the other hand, we prove that for almost all fixed-weight boolean functions, the upper bound is tight up to an additive logarithmic term. We show that, for any $w \ge 2\log n$, there exists a DNF formula of width $w$ and size $\lceil2^w / w\rceil$ such that $\avgq(f) = n (1 - \frac{\log n}{\Theta(w)})$, which suggests that the criticality bounds $O(w)$ and $O(\log s)$ are tight up to a multiplicative $\log n$ factor.

Theorems \ref{thm:intro_bnm_ub} and \ref{thm:intro_bnm_lb} essentially relate \(\avgq(f)\) to the zero-order Fourier coefficient \(\widehat{f}(\{\emptyset\})\). 
Establishing an upper bound on $\avgq(f)$ in terms of higher-order Fourier coefficients (such as influences) would be valuable. 
For example, it is unclear whether the lower bound $\avgq(\mathrm{Hex}_{L \times L}) \ge L^{1.5 + o(1)}$ is tight \cite{PSS+07}; bounding $\avgq(f)$ in terms of Fourier coefficients might shed light on the open problem.

It remains open to prove tight upper bounds on $\avgq(f)$ for $k$-DNF, as well as for bounded depth formulas and circuits.

\subsection*{Acknowledgements}

We are grateful to the anonymous reviewers for their valuable feedback.

\printbibliography

\appendix

\begin{appendices}

\section{Penalty shoot-out function} \label{apdx:pso}

Besides the AND/OR functions, which are highly biased, there are monotone balanced functions such that the gap between $\avgq(f)$ and $\detq(f)$ is arbitrarily large. 

Let us define the \emph{penalty shoot-out function} $\pso_n: \bitset^{2n+1} \rightarrow \bitset$ as follows. Consider an $n$-round penalty shoot-out in a football game.  In each round, two teams, A and B, each take a penalty kick in turn, with team A going first. 
Let $x_{2i - 1} = 1$ indicate the event team A scores in the $i$-th round, and let $x_{2i} = 0$ --- this is to make the function \emph{monotone} --- indicate the event that team B scores in the $i$-th round for $i = 1, 2, \ldots, n$. 
The game continues until one team scores \emph{and} the other does not (within the same round). 
If no winner is declared after $2n$ kicks, an additional kick by team A decides the game. In this final kick, if team A scores, team A wins and $\pso_n(x) = 1$; otherwise, team B wins and $\pso_n(x) = 0$.  To the best of our knowledge, $\pso_n$ is first studied here.

Assume both teams have equal probabilities of scoring, that is, $\pso_n$ is defined under a uniform distribution. The function $\pso_n$ is a monotone balanced function with $\avgq(\pso_n) = O(1)$ and $\detq(\pso_n) = \Theta(n)$. 

\begin{proposition} \label{prop:pso-dave}
    $\detq(\pso_n) = 2n+1$ and $\avgq(\pso_n) = 4 - \frac{3}{2^n}$. 
\end{proposition}

\begin{proof}
In the worst case, the winner cannot be declared until the last round is finished, so $\detq(\pso_n) = 2n+1$.
    
We prove $\avgq(\pso_n) = 4 - \frac{3}{2^n}$ by induction on $n$. When $n = 0$, $\avgq(\pso_n) = 1$, because one kick by team A decides the game.
Assuming $\avgq(\pso_{n-1}) = 4 - \frac{3}{2^{n-1}}$ holds, let us prove $\avgq(\pso_n) = 4 - \frac{3}{2^n}$. In each round, there are four cases:
    \begin{itemize}
        \item[(1)] Team A scores, and team B does not.
        \item[(2)] Team B scores, and team A does not.
        \item[(3)] Both teams score.
        \item[(4)] Neither scores.
    \end{itemize}
    Each case happens with equal probability $\frac{1}{4}$. If case 1) or 2) happens, the winner is decided; if case 3) or 4) happens, the game will continue. Therefore, we have
    \begin{align*}
        \avgq(\pso_{n}) & = \frac{1}{2} \cdot 2 + \frac{1}{2} \cdot \left(\avgq(\pso_{n-1}) + 2\right) \\
        & = \frac{1}{2} \left(4 - \frac{3}{2^{n-1}}\right) + 2\\
        & = 4 - \frac{3}{2^n},
    \end{align*}
    completing the induction proof. 
\end{proof}

\end{appendices}

\end{document}